\documentclass[letterpaper,prl,twocolumn,showpacs]{revtex4}

\usepackage{amsmath,amssymb}
\usepackage{epsfig}
\usepackage{amsthm}
\usepackage{graphicx}
\usepackage{float}
\usepackage[caption = false]{subfig}
\usepackage{bm}
\usepackage{color}

\newcommand{\sign}{ \mbox{sign}}
\newcommand{\sig}{ \mbox{\boldmath{$\sigma$}}}
\newcommand{\Spec}{\mathop{\mathrm{Spec}}\nolimits}

\newtheorem*{SpectrumT}{Theorem}
\theoremstyle{plain}

\newtheorem*{Lev}{Theorem}
\theoremstyle{plain}

\newtheorem*{Poincare}{Theorem}
\theoremstyle{plain}

\newtheorem{lemma}{Lemma}
\theoremstyle{plain}

\newtheorem{remark1}{Remark}
\theoremstyle{remark}

\begin{document}

\title{Analytical study of bound states in graphene nano-ribbons and carbon nanotubes: the variable phase method and the relativistic Levinson theorem}

\date{\today}

\author{D.~S.~Miserev$^{1, 2}$} \email{d.miserev@student.unsw.edu.au}
\affiliation{$^{1}$School of Physics, University of New South Wales, Sydney, Australia}
\affiliation{$^{2}$Rzhanov Institute of Semiconductor Physics, Siberian Branch, Russian Academy of Sciences,
pr. Akademika Lavrent’eva 13, Novosibirsk, 630090 Russia}

\begin{abstract}
The problem of localized states in 1D systems with the relativistic spectrum, namely, graphene stripes and carbon nanotubes, has been analytically studied. The bound state as a superposition of two chiral states is completely described by their relative phase which is the foundation of the variable phase method (VPM) developed herein. Basing on our VPM, we formulate and prove the relativistic Levinson theorem. The problem of bound state can be reduced to the analysis of closed trajectories of some vector field. Remarkably, the Levinson theorem appears as the Poincare indices theorem for these closed trajectories. The reduction of the VPM equation to the non-relativistic and semi-classical limits has been done. The limit of the small momentum $p_y$ of the transverse quantization is applicable to arbitrary integrable potential. In this case the only confined mode is predicted.
\end{abstract}

\maketitle 
\section{Introduction}
Graphene, carbon nanotubes and topological insulators have attracted keen attention for intensive theoretical and experimental research in recent years. The uniqueness of these quantum materials with respect to fundamental physics lies in the opportunity to observe QED effects with a significantly larger coupling constant $g = e^2/s \hbar \varepsilon \sim 1$, where $s \approx c/300$ is the Fermi velocity, $\varepsilon$ is an average dielectric constant of environment (for instance, for graphene sheet on the substrate with the dielectric constant $\varepsilon_s$ one obtains $\varepsilon = (1 + \varepsilon_s)/2$). Effects such as the atomic collapse and pair production in the super-critical potentials~\cite{Popov}--\cite{Milstein}, the Adler-Bell-Jackiw anomaly (the chiral anomaly)~\cite{Nielsen}--\cite{Landsteiner} have been intensively studied. The Klein tunnelling of electrons in the gated graphene~\cite{Ando}--\cite{Reijnders} reveals the complete suppression of the backscattering. 

The present work is related to the general theoretical study of the confined electronic states in graphene nano-ribbons or single-walled carbon nanotubes affected by a longitudinal electric field. Omitting inter-valley scattering, we consider electron behavior near one of two independent Dirac points where electrons are well-described by the Dirac-Weyl hamiltonian~(\ref{hamiltonian}) in the single-particle approach. 

We propose a convenient technique to analyse bound states analytically for the 2D Dirac-Weyl equation with a 1D potential $U(x)$. It refers to the variable phase method (VPM) developed generally by P. M. Morse and W. P. Allis~\cite{Morse}, V. V. Babikov~\cite{Babikov}, F. Calogero~\cite{Calogero} and others~\cite{Sobel}--\cite{Ouerdane}. The wave function is expressed as a linear combination of two Weyl fermions and the phase between them is considered as a desired phase function for the VPM to be applied. Following this, we demonstrate the reduction to the non-relativistic and semi-classical limits. Furthermore, we consider one more limiting case of the $\delta$-potential which is applicable to any integrable potentials at sufficiently small transverse momentum $p_y$. Physically, this limit contains both the shallow quantum well limit and the opposite limit of a strongly supercritical potential.

Our VPM allows one to formulate the relativistic analogue of the Levinson theorem~\cite{Levinson}. The relativistic Levinson theorem for the Dirac equation was formulated in 3D by M. Klaus~\cite{Klaus} for central potentials, K. Hayashi~\cite{Hayashi} and R. L. Warnock~\cite{Warnock} as a relation between zeroes of the vertex function and particle poles of the total amplitude. This problem has been considered in two dimensions with the compact supported central potential~\cite{Dong}. D. P. Clemence~\cite{Clemence} thoroughly investigated the Levinson theorem for the Dirac equation with a 1D potential which satisfies the condition $\int_{-\infty}^{\infty} U(x) (1 + |x|) \, dx < \infty$ via the scattering matrix approach taking into account the half-bound states. The particular case of the relativistic Levinson theorem for symmetric 1D potentials has been studied by Q. Lin~\cite{Lin} with additional restriction for the potential to be a compact supported function, A. Calogeracos and N. Dombey~\cite{Calogeracos} for potentials of definite sign, Z. Ma {\it et al.}~\cite{Ma} with the similar condition as in~\cite{Clemence}. The developed herein method permits one to prove the Levinson theorem with the minimal restriction $\int_{-\infty}^\infty U(x) \, dx < \infty$ which significantly broadens the result obtained by D. P. Clemence. For example, our results are applicable to so-called top-gate potential~(\ref{top-gate}) for which asymptotics are expected to be realistic for the gated graphene structures~\cite{Hartmann}. Afterwards, a geometrical interpretation of the Levinson theorem together with the corresponding numerical method of integral curves analysis of some vector field are considered.

\section{Theoretical model}
Near the conic points, electrons in graphene with the gated potential $U(x)$ are approximately described by the Dirac-Weyl Hamiltonian:
\begin{equation}
\hat H = s {\sig \hat{\bf p}} + U(x)= s \sigma_x \hat p_x +s \sigma_y \hat p_y + U(x)
\label{hamiltonian}
\end{equation}
where $s$ is the Fermi velocity, $\sig = (\sigma_x, \sigma_y)$ are Pauli matrices, ${\bf p}  = - i \hbar \bm{\nabla}$. Henceforth, it is assumed that the potential decays at infinity. Further calculations are executed in the dimensionless variables: $\hbar = s = 1$. It is also assumed $p_y > 0$ where $p_y$ is the quantized transverse momentum of quasi-1D systems such as graphene nano-ribbons and single-walled carbon nanotubes where $y = r \phi$, $r$ is the radius, $\phi$ is the cyclic variable. The spectrum of the free-particle Hamiltonian is linear on the momentum: $E = \pm \sqrt{p_x^2+p_y^2}$. The negative-energy states correspond to the hole's description according to the conventional views.

The stationary wave function can be represented in a symmetric form:
\begin{equation}
\label{psi}
\Psi = \frac{ e^{i p_y y}}{\sqrt{4 W}}{g (x)+p_y^{-1} g'(x)\choose g (x)-p_y^{-1} g'(x)}e^{i \int \limits^x\left(E-U(\zeta)\right)\, d\zeta}
\end{equation}
via the axillary function $g\left( x \right)$ which is introduced in~\cite{Miserev}:
\begin{equation}
\label{eq}
g''\left(x\right)+2i\left(E-U\left(x\right)\right)g'\left(x\right)-p_y^2 g\left(x\right)=0
\end{equation}
where $E$ is the electron energy and $W$ is the normalization coefficient. Eq.~(\ref{eq}) represents an equivalent statement of the problem described by the Hamiltonian~(\ref{hamiltonian}). Further we deal with electronic states of zero current along $x$-direction. 

We now apply this condition to the analysis of confined states. Zero flow $j_x=\Psi^\dag(x) \sigma_x \Psi(x) = 0$ along $x$-direction yields the restriction on the function $g(x)$:
\begin{equation}
\label{feature1}
|g\left(x\right)|=|p_y^{-1} g'\left(x\right)|.
\end{equation}

The first consequence is that $g(x)$ and hence the electron density of confined states $\rho(x) = \Psi^\dag(x) \Psi(x) = |g(x)|^2/W$ vanishes only at infinity. Otherwise, we have from~(\ref{feature1}): $g(x_0) = g'(x_0) = 0$, $|x_0| < \infty$, which yields $g(x) \equiv 0$.

Separating modulus and phase $g(x)=R e^{i\Phi}$, we arrive at the condition: 
\begin{equation}
\label{pifagor}
\left(\Phi'\right)^2+\left(R'/R\right)^2=p_y^2,
\end{equation}
which allows for the following substitution:
\begin{equation}
\label{newfunc}
\left\{
\begin{array}{rcl}
\Phi'(x)&=&p_y \sin {\Omega(x)}\\
R'/R&=&p_y \cos {\Omega(x)}
\end{array}
\right.
\end{equation}
where the function $\Omega(x)$ is the solution of the first-order differential equation:
\begin{equation}
\label{supereq}
\Omega'(x)=2 \left(U(x)- E\right) -2 p_y \sin {\Omega(x)}.
\end{equation}
Thereby, we arrived at the desired VPM equation. We emphasize here that Eq.~(\ref{supereq}) is valid for any quantum state with zero flow, not only for bound states.

Considering bound states, we have to set the boundary conditions for the function $\Omega(x)$:
\begin{equation}
\label{asymptotes}
\left\{
\begin{array}{ll}
\Omega(x \to + \infty)=\pi +\arcsin {\frac{E}{p_y}}+2 \pi n\\
\Omega(x \to - \infty)=-\arcsin {\frac{E}{p_y}}.
\end{array}
\right.
\end{equation}
At $E \in(-p_y, p_y)$ these conditions provide the exponential decay of the density $\rho(x) \sim R^2(x)$ at infinity as it follows from~(\ref{newfunc}), $n$ being an integer.

To reveal the physical meaning of the function $\Omega(x)$, we use the following representation of the wave function:
\begin{equation}
\label{psidiscrete}
\Psi(x, y)= \frac{ e^{i p_y y}}{\sqrt{4 W}}\left({1 \choose 1} + e^{i \Omega} {1 \choose -1}\right) R(x) e^{-i \Omega/2}.
\end{equation}
Hence, confined state appears as a linear combination of two chiral (Weyl) states and is completely described by the phase between them. 
Another form of Eq.~(\ref{psidiscrete}) refers to the spin with the polar angle  $\Omega$  and the azimuthal angle $-\pi/2$:
\begin{equation}
\label {psidiscrete2}
\Psi (x, y)= \frac{R(x) e^{i p_y y}}{\sqrt{W}}{\phantom{-i}\cos \frac{\Omega}{2}\choose -i \sin \frac{\Omega}{2}}.
\end{equation}

\section{Non-relativistic limit}
Let us show that Eq.~(\ref{supereq}) can be reduced to a non-relativistic equation. 
To be more specific, consider the non-relativistic limit for electrons: 
\begin{eqnarray*}
E = p_y + \varepsilon, \\
\varepsilon = - k^2 /2 p_y,
\end{eqnarray*}
where we imply that all energy scales are small as compared with $p_y$: $k, U(x), 1/d \ll p_y$, $d$ is the characteristic width of the confinement. Boundary conditions~(\ref{asymptotes}) for $\Omega(x)$ take the form: $\Omega(-\infty) = -\pi/2 + k/p_y$, $\Omega(+\infty) = -\pi/2 - k/p_y +2 \pi n$, $n$ being an integer. 

Suppose $\Omega(x) = -\pi/2 + \delta \Omega$, where $\delta \Omega \ll 1$ almost everywhere. This assumption is violated only when $\Omega' \sim p_y$ which corresponds to $\delta \Omega \sim 1$. The behaviour of the phase function $\Omega(x)$ in this region does not depend on the potential because $U(x) \ll p_y$. Notice that the width of this region $\delta x \sim 1/p_y \ll d$ is small in the non-relativistic limit. Hence, the expansion of the initial equation~(\ref{supereq}) results in the Riccati equation:
\begin{equation}
\label{NonRel}
\delta \Omega' = 2(U(x) - \varepsilon) - p_y \delta \Omega^2,
\end{equation}
where $\psi(x) = \exp \left( p_y \int \delta \Omega(x) \, dx \right)$ satisfies the 1D Schrodinger equation for a non-relativistic particle with mass $p_y$. The function $\delta \Omega(x)$ tends to the infinity in zeroes of the wave function $\psi(x)$.

\section{Semi-classical limit}
Let us rewrite Eq.~(\ref{supereq}) in the dimensional quantities:
\begin{equation}
\hbar \Omega' = \frac{2}{s} \left( U(x) - E\right) - 2 p_y \sin \Omega,
\end{equation}
where $s$ is the Fermi velocity.
In the semi-classical limit $\hbar \to 0$ the elimination of the left-hand part of this equation yields:
\begin{equation}
\label{quasiw}
\sin \Omega = \frac{U(x) - E}{s p_y}.
\end{equation}
Let us show that Eq.~(\ref{quasiw}) represents the usual quasi-classical approach.

This approximation is solvable in the real-valued functions when $|U(x) - E| < s p_y$, which conforms to the case of non-classical motion where the wave function decays. At breakpoints $x_i$, when $U(x_i) - E = -\mu \cdot s p_y$ we define $\Omega(x_i) = - \mu \pi/2$, $\mu = \pm 1$ is definite for each region of motion. 

In the regions of classical motion where the wave function is oscillatory, $\Omega(x)$ is a complex function, namely, $\Omega(x) = -\mu \pi/2 + i \delta \Omega$:
\begin{equation}
\label{domega}
\cosh \delta \Omega(x) = -\mu \frac{U(x) - E}{s p_y} = \left| \frac{U(x) - E}{s p_y} \right|.
\end{equation} 
Eq.~(\ref{domega}) has two solutions $\pm \delta \Omega$ (for definiteness, we set the first solution $\delta\Omega \ge 0$).
The corresponding amplitude of the wave function $R_\pm(x)$ is determined from Eq.~(\ref{newfunc}):
$$
R_\pm (x) \sim \exp\left(\pm i \frac{p_y}{\hbar} \int \sinh \delta \Omega(x) \,dx \right).
$$
According to the definition, it is required that the function $R(x)$ is real-valued. It means that we have to consider a linear combination of corresponding functions $g_\pm (x) = R_\pm (x) e^{i \Phi_\pm(x)}$ where
$$
\Phi_\pm (x) = -\mu \displaystyle \int \left| \frac{U(x) - E}{s}\right| \, \frac{dx}{\hbar} = \displaystyle \int \frac{U(x) - E}{s} \, \frac{dx}{\hbar},
$$
which follows from Eq.~(\ref{newfunc}) and $\Phi$ is the same for the two different solutions of Eq.~(\ref{domega}).
Finally, the semi-classical amplitude reads:
\begin{equation}
R(x) \sim \cos\left(\int p_x \, \frac{dx}{\hbar} + \phi_0 \right)
\end{equation}
where the semi-classical momentum $p_x= p_y  \sinh \delta \Omega(x)= \sqrt{\left( E - U(x) \right)^2/s^2 - p_y^2}$ is introduced. The phase $\phi_0$ is defined by the matching conditions. 

Hence, Bohr-Sommerfeld quantization takes the usual form:
\begin{equation}
\oint p_x dx = 2 \pi \hbar (n + \gamma)
\end{equation}
where $n \gg 1$ is an integer, $\gamma \sim 1$ is defined from the matching conditions in the turning points; for example, $\gamma = 1/2$ for smooth potentials. The semi-classical approximation is valid when $\hbar p_y U'(x) \ll s p_x^3$.

\section{Delta-potential limit}
\label{Deltap}
Before we start, we emphasize that we do not require from the confinement $U(x)$ to be $\delta$-like. The reason why we name this limit as the delta-potential limit is that at some conditions the discrete spectrum and corresponding wave functions of any integrable potential are of the same analytical form as for the actual $\delta$-potential which is considered in Appendix A.

In this section we are interested in all possible cases when we are entitled to neglect the non-linear term in Eq.~(\ref{supereq}). It allows to find the spectrum and corresponding wave functions exactly. Let us formulate the following 

\begin{SpectrumT}
Let the potential $U(x)$ be an integrable function, $d$ is the characteristic width of $U(x)$, $p_y > 0$ is transverse momentum. Introduce the integral
	\begin{equation}
		G = \int_{-\infty}^\infty U(x) \, dx = \pi (n_G + \delta n_G),
		\label{gexp}
	\end{equation}
where $n_G$ is integer and $\delta n_G \in [0, 1)$
Assume $\delta n_G \ne 0$. 

Let the condition be met:
\begin{equation}
p_y d \ll \min \{\delta n_G, 1 - \delta n_G\}.
\label{pycond}
\end{equation}
Then:
\begin{itemize}
	\item[{\bf a}]
	The discrete spectrum contains the only one level with energy $E \in (-p_y, p_y)$:
	\begin{equation}
	\label{Deltalimit}
	E = (-1)^{n_G+1} p_y \cos G,
	\end{equation}
	
	\item[{\bf b}]
	If additionally $\int\limits_{x_0}^{x} U(x') x' \, dx'$ converges at $x \to \pm \infty$ at some $|x_0| < \infty$, the corresponding wave function takes the form~(\ref{psidiscrete2}) with the phase function:
	\begin{equation}
	\Omega(x) = -\arcsin \frac{E}{p_y} + 2 \int \limits_{-\infty}^x U(x') \, dx'.
	\label{wavef}
	\end{equation}
	
	\end{itemize}
\end{SpectrumT} 

\begin{proof}[{\bf Proof.}] 
	
	We mean here that $U(x)$ is an integrable function in a sense that the primitive integral 
	$$
	f_{x_0}(x) = \int_{x_0}^x U(x') \, dx'
	$$ 
	for some $|x_0| < \infty$ is defined for any $x \in (-\infty, +\infty)$ except maybe some finite set of points, and $f_{x_0}(x)$ is bounded function. We set parameter $E \in(-p_y, p_y)$.

\begin{itemize}
	\item
	Let $\Omega(x)$ is a physical solution with boundary conditions~(\ref{asymptotes}).
	Then the total variance of the phase function $\Delta\Omega = \Omega(+\infty) - \Omega(-\infty)$ is straightforward from~(\ref{asymptotes}):
	\begin{equation}
		\Delta \Omega = 2 \arcsin \frac{E}{p_y} + 2 \pi \left(n+\frac{1}{2}\right).
		\label{omvar}
	\end{equation}
	
	On the other hand, the integration of Eq.~(\ref{supereq}) yields:
	\begin{equation}
	\label{DeltaOmega}
	\Delta \Omega = 2 G + \mathfrak{K},
	\end{equation}
	where $n$ is the integer. We introduced the integral:
	\begin{equation}
		\mathfrak{K} = \int \limits_{-\infty}^\infty 2 (E + p_y \sin \Omega(x))\, dx.
	\end{equation}
	
	{\it Convergence of $\mathfrak{K}$.}
	
	Let us use {\bf Lemma 2} about the properties of solutions of Eq.~(\ref{supereq}) and rewrite $\mathfrak{K}$:
	$$
	\mathfrak{K} = 2 p_y \int\limits_{-\infty}^\infty(\sin \Omega(x) - \sin \Omega_\pm)\, dx.
	$$
	From {\bf Lemma 2} we know that the physical solution corresponds to the degeneration of two separatrix families of Eq.~(\ref{supereq}). Let us consider the behavior of this physical solution at $x \to -\infty$ where we can represent it in the form:
	$$
	\Omega(x) = \Omega_- + \delta\Omega(x).
	$$
	At $x \to -\infty$, $\delta\Omega(x)$ satisfies the approximate equation which follows directly from Eq.~(\ref{supereq}):
	$$
	\delta\Omega'(x) \approx 2 U(x) - 2 k \cdot \delta\Omega(x),
	$$
	where we accounted for that $p_y \cos\Omega_- = k > 0$, $k = \sqrt{p_y^2 - E^2}$. The solution which meets the initial condition $\delta\Omega(-\infty) = 0$ reads:
	\begin{equation}
	\delta\Omega(x) = 2 \int\limits_{-\infty}^{x} U(x') e^{-2k (x-x')} \, dx'.
	\label{domex}
	\end{equation}
	
	Apply it to analyze the convergence of $\mathfrak{K}$ at $-\infty$. If $x \to -\infty$ we can use the expansion $p_y (\sin\Omega(x) - \sin\Omega_-) \approx k \cdot \delta\Omega(x)$. Then we get:
	\begin{eqnarray*}
	&& 2 p_y \int\limits_{-\infty}^x (\sin\Omega(x') -\sin\Omega_-) \, dx' \approx  \\ 
	&& \approx 2 k \int\limits_{-\infty}^x \delta\Omega(x') \, dx' = 2 \int\limits_{-\infty}^x U(x') \, dx' - \delta \Omega(x).
	\end{eqnarray*}
	It proves the convergence of $\mathfrak{K}$ at $-\infty$ once $U(x)$ is an integrable function. One can prove by analogy the convergence at $+\infty$. Hence, $\mathfrak{K}$ converges. 
	
	{\it Estimation of $\mathfrak{K}$.}
	
	The convergence allows us to introduce some characteristic scale $D(\varepsilon)$ which is a diameter of the convergence domain of $\mathfrak{K}$. Mathematically, for any $\varepsilon > 0$ the number $0 < D(\varepsilon) < \infty$ exists that 
	$$
		\left|\mathfrak{K} - 2 p_y \int\limits_{-D(\varepsilon)/2}^{D(\varepsilon)/2} (\sin \Omega(x) - \sin \Omega_\pm)\, dx \right| < \varepsilon.
	$$

	We will consider only those cases when we can omit $\mathfrak{K}$ in Eq.~(\ref{DeltaOmega}). Then, let us estimate the order of magnitude. As we can see from the convergence proof, integrals $\mathfrak{K}$ and $G$ converge simultaneously. Then:
	\begin{equation}
		\mathfrak{K} \sim O(p_y \cdot d),
		\label{Kestim}
	\end{equation}
	where $d$ is the characteristic convergence length of the integral $G$ or, alternatively, the characteristic length of the confinement.
	
	We are ready now to prove the theorem. 
	
	\item[{\bf a}]
	Combining Eq.~(\ref{omvar}) and Eq.~(\ref{DeltaOmega}) we get:
		\begin{equation}
			\label{spec}
			\arcsin \frac{E}{p_y} =  \pi\left( \delta n_G + \frac{\mathfrak{K}}{2\pi} - \frac{1}{2} + n_G-n\right).
		\end{equation}
	If the condition~(\ref{pycond}) is met, we can omit $\mathfrak{K}$ in Eq.~(\ref{spec}). After that we can set $n = n_G$ because $\arcsin x \in [-\pi/2, \pi/2]$ which finally gives:
	$$
	 \arcsin \frac{E}{p_y} = \pi \left(\delta n_G - \frac{1}{2}\right)
	$$
	that is equivalent to Eq.~(\ref{Deltalimit}).
	
	\item[{\bf b}] 
	In order to obtain the wave function, we can naively neglect the influence of the non-linear term of Eq.~(\ref{supereq}) and, hence, the approximate solution reads:
	$$
	\Omega_0(x) = \Omega_- + 2\int\limits_{-\infty}^x U(x') \, dx'
	$$
	which coincides with~(\ref{wavef}).
	However, this approximation is valid when there is no divergence in the following correction of order of $p_y d$. This correction can be estimated as follows:
	\begin{eqnarray*}
	&&\Omega_1(x) = \\
	&&= -2 p_y \int\limits_{-\infty}^x (\sin\Omega_0(x) - \sin\Omega_-) \, dx' + \Omega_1(-\infty),
	\end{eqnarray*}
	where we imply that the integral converges. Checking the convergence at $x \to -\infty$:
	\begin{eqnarray*}
		&&\Omega_1(x) \approx -2 k \int\limits_{-\infty}^x \int\limits_{-\infty}^{x'} U(x'') \, dx'\, dx'' + \Omega_1(-\infty),
	\end{eqnarray*}
	where this double integral reduces to $\int\limits_{-\infty}^x U(x') x' \, dx'$, which means that we can use the approximate wave function~(\ref{wavef}) only when $x U(x)$ is integrable.
	
	This is unsurprising because for the convergence of $\mathfrak{K}$ at the condition of integrability of $U(x)$ we required the exponential decay of $\Omega(x)$ to $\Omega_-$ at $x \to -\infty$ as it is shown by Eq.~(\ref{domex}). It means that we cannot neglect the dependence of wave function on $k$ and thus, we are not allowed to use the approximate wave function~(\ref{wavef}) if $U(x)$ is integrable but not $x U(x)$. However, the spectrum~(\ref{Deltalimit}) is valid even if $x U(x)$ is non-integrable once $U(x)$ is integrable and the condition~(\ref{pycond}) is met.
		
	Physically, this limit can be understood as a supercritical regime for the confinement $U(x)$. If we consider the case where $U(x)$ is a quantum well with the characteristic depth $U_0$ and width $d$, then, $\pi \delta n_G \lesssim G \sim U_0 \cdot d$ and the condition~(\ref{pycond}) gives $U_0 \gg p_y$ which corresponds to the strong supercritical regime. 

	Hence, once the condition~(\ref{pycond}) is valid, we get for any integrable potential:
	\begin{equation}
	\label{DeltaLimit}
	\arcsin \frac{E}{p_y} \approx  G - \pi \left(n+\frac{1}{2}\right).
	\end{equation}
\end{itemize}
\end{proof}

We did not consider the cases $G = \pi n_G$, $n_G$ is an integer because it requires more fine analysis than represented above.

{\it Zero-energy states}

We are going to compare our results with some recent analytical works on graphene states. As an example, let us consider the condition for the existence of confined modes with zero energy (exactly in Dirac point). Zero-energy confined states and their importance in possible construction of 1D gated structures (waveguides) were discussed thoroughly in~\cite{Hartmann}. 

According to Eq.~(\ref{DeltaLimit}), we arrive at the desirable restriction, if Eq.~(\ref{pycond}) is valid:
\begin{equation}
\label{Zeromodes}
G =\pi \left( n + \frac{1}{2} \right),
\end{equation}
where $n$ is an integer. This constriction means that we cannot have zero-energy confined states at arbitrarily small potential strength $G$. However, at any $G \ne \pi n$ we have at least one bound state.

In~\cite{Hartmann} the analytical solution for zero-energy modes in the gate potential $V(x) = - U_0 /\cosh(x/d)$, $U_0 > 0$, is provided. 
Taking into account that for this case $G = - \pi U_0 d$ we arrive at the condition for zero-energy mode existence in the limit of small $p_y$:
$$
U_0 d = n + \frac{1}{2}
$$
where $n$ is a non-negative integer. Hence, we cannot have a confined zero-energy modes once $|U_0 d| < 1/2$ which coincides exactly with the condition obtained analytically in~\cite{Hartmann}.

Thorough analytical study of bound states in the potential 
\begin{equation}
	V(x) = - U_0 /\cosh(x/d)
	\label{secant}
\end{equation}
for non-zero energies has been done in the recent paper~\cite{Quasi}. The authors claim that there is a threshold value of the potential strength $G = \pi U_0 d > \pi/2$ for the first confined state to appear. We suppose that something essential is missing in the work~\cite{Quasi} since this strong statement immediately contradicts the non-relativistic limit and the limit of $\delta$-potential that are developed herein. 

Let us now compare our VPM method with one developed by D. A. Stone {\it et al.}~\cite{Portnoi}. They considered another phase function which satisfies a more complex equation. One of the substantial points of their paper is that zero-energy mode exists for arbitrarily small power-law decaying (faster than $1/x$) potentials. And again this statement strongly contradicts with Eq.~(\ref{Zeromodes}). Moreover, their asymptotic analysis resulted in no bound states for the potential~(\ref{secant}) if  $p_y < 1/d$. It apparently contradicts with our $\delta$-limit. 

Finally, consider the potential $V(x) = U_0 \exp(- |x|/d)$. Zero-energy mode condition was found analytically in~\cite{Portnoi} where the minimal potential strength is stated as $(U_0 d)_{min} = \pi/4$. Our model predicts zero-energy modes when $2 U_0 d = \pi (n+1/2)$ in excellent agreement with analytical solution.

Due to the simplicity of our method, let us calculate the condition of zero-energy mode existence for so-called top-gate potential $V_{t}(x)$ (see reference~\cite{Hartmann}):
\begin{equation}
\label{top-gate}
V_t(x) = \displaystyle \frac{U_0}{2} \ln \left(\frac{x^2 + (h_2 - h_1)^2}{x^2 + (h_2 + h_1)^2}\right)
\end{equation}
where parameters $h_1 < h_2$ depend on geometry of the gate electrodes. Namely, $h_1$ is a width of the insulator between the graphene plane and so-called back-gate electrode, $h_2$ is a distance between top and back electrodes. Applying Eq.~(\ref{Zeromodes}) one receives the condition of zero mode existence:
$$
U_0 h_1 = \frac{1}{2} \left( n + \frac{1}{2}\right) \ge \frac{1}{4}.
$$
Notice that this condition does not depend on the bigger parameter $h_2$ which in our case determines the distance between electrodes. 

Hence, the $\delta$-potential limit is a simple and powerful tool to study one-particle confined states in arbitrary integrable 1D gate potentials in graphene stripes and it should be included in the analysis of bound states for concrete configuration of the gate potential to avoid possible misconceptions.

\section{Relativistic Levinson theorem}

In this section, we formulate the oscillation theorem in terms of the phase function $\Omega(x)$ as it has been done for the case of massive non-relativistic particles through the analysis of the scattering phase function~\cite{Morse}.

Before we set out the main theorem, we give some properties of the solutions to Eq.~(\ref{supereq}).

\begin{lemma}[of continuity]
	Define the following function: $f_{x_0}(x) =\int_{x_0}^{x} U(x') \, dx'$, $|x_0| < \infty$ is some constant. Let $f_{x_0}(x) \in C^k$, where $k$ is a non-negative integer, $C^k$ is the $k$-th class of differentiability. Then every solution of Eq.~(\ref{supereq}) belongs to $C^k$.
\end{lemma}
\begin{proof}[{\bf Proof}]
	We prove this by induction.
	\begin{itemize}
		\item[{\bf a}] If $k=0$ then $f_{x_0}(x)$ is a continuous function. It is equivalent to the condition:
		$ \int_{x}^{x+\epsilon} U(x') \, dx' \to 0$ if $\epsilon \to 0$ at arbitrary $x \in (-\infty, \infty)$.
		Then, integrate Eq.~(\ref{supereq}) from $x$ to $x+\epsilon$:
		\begin{eqnarray*}
			|\Omega(x+\epsilon) - \Omega(x)| = \left| 2\int_{x}^{x+\epsilon} U(x') \, dx' - \right.\\
			\left.-2 \int_{x}^{x+\epsilon} (E + p_y \sin \Omega(x')) \, dx'\right| \le \\
			2 \left|\int_{x}^{x+\epsilon} U(x') \, dx' \right| + 2 \epsilon (p_y +|E|) \to 0,
		\end{eqnarray*}
		which confirms the continuity of any solution of Eq.~(\ref{supereq}).
		
		\item[{\bf b}]  Assume that the statement of the lemma is true at all $k<n$, where $n$ is positive integer. Let $f_{x_0}(x) \in C^n$. Then prove the Lemma at $k=n$. Differentiate Eq.~(\ref{supereq}) $n-1$ times:
		$$
		\Omega^{(n)}(x) = f_{x_0}^{(n)}(x) - 2(E+ p_y \sin\Omega(x))^{(n-1)},
		$$
		where $f_{x_0}^{(n)}(x)$ is continuous by the condition of the lemma. $2(E+ p_y \sin\Omega(x))^{(n-1)}$ is continuous by inductive assumption because it contains derivatives of $\Omega(x)$ not higher than $n-1$. Then $\Omega^{(n)}(x)$ is continuous function, or $\Omega(x) \in C^{(n)}$.
	\end{itemize}
\end{proof}

We need to make one additional comment. If $f_{x_0}(x)$ is a piecewise-continuous function (this means that $U(x)$ has $\delta$-like singularities at discontinuity points), all solutions of Eq.~(\ref{supereq}) are piecewise-continuous with the same discontinuity points as $f_{x_0}(x)$. In other words, the statement of the {\bf Lemma 1} is valid even if $f_{x_0}(x)$ is a piecewise-continuous function.

\begin{lemma}[of attractors and repellors]
	Let $U(x) \to 0$ at $x \to \infty$, $E\in (-p_y, p_y)$. Then:
	\begin{itemize}
		\item[{\bf a}] All solutions of Eq.~(\ref{supereq}) at infinity come to stationary points of the free motion equation (i.e. with zero potential).
		\item[{\bf b}] There are two families of stationary points:
		\begin{equation}
		\label{attract}
		\left\{
		\begin{array}{ll}
		\Omega_- = -\arcsin (E/p_y) + 2 \pi n\\
		\Omega_+ = \arcsin (E/p_y) + 2 \pi \left(n + 1/2\right).
		\end{array}
		\right.
		\end{equation}
		\item[{\bf c}] $\Omega_+$ ($\Omega_-$) is an attractor (repellor) at $x\to -\infty$; \\
		$\Omega_+$ ($\Omega_-$) is a repellor (attractor) at $x\to +\infty$.
		\item[{\bf d}] There are two types of separatrix solutions which are defined by following Cauchy problems:
		\begin{equation}
		\label{separatrix}
		\left\{
		\begin{array}{ll}
		\Omega_l(x \to -\infty)=\Omega_-\\
		\Omega_r(x \to +\infty)=\Omega_+.
		\end{array}
		\right.
		\end{equation}
		We call $\Omega_l(x)$ ($\Omega_r(x)$) the left (right) separatrix.
		\item[{\bf e}] The bound state problem is equivalent to the degeneracy of two separatrix families $\Omega_l$ and $\Omega_r$.
	\end{itemize}
\end{lemma}

\begin{proof}[{\bf Proof}]
	\begin{itemize}
		\item[{\bf a}]
				
		Consider the free motion equation:
		\begin{equation}
		\Omega'(x) = -2 p_y \left(\sin \Omega(x) + \frac{E}{p_y} \right).
		\label{freemotion}
		\end{equation}
		This equation has stationary points $\Omega(x) \equiv const$ when $\sin \Omega = -E/p_y$. Every solution of Eq.~(\ref{freemotion}) comes to $\Omega_+$ ($\Omega_-$) at $x \to -\infty$ ($x \to +\infty$), where $\Omega_\pm$ are defined according to~(\ref{attract}). Moreover, $\Omega_\pm$ are solutions by itself. However, there are no physical solutions amid the solutions of the free motion equation because it is impossible to satisfy physical boundary conditions~(\ref{asymptotes}).
		
		If we have $U(x) \to 0$, $x \to \infty$, asymptotics of solutions at infinity resemble those of the free motion equation. Thus, {\bf a} is proven.
		
		\item[{\bf b}]
		Two families of stationary points of the free motion equation (which present the whole set of attractors and repellors of Eq.~(\ref{supereq})) obviously arise from the equation $\sin \Omega_\pm = -E/p_y$.
		
		\item[{\bf c}]
		Let us demonstrate that $\Omega_+$ are repellors at $x \to +\infty$ and attractors at $x \to - \infty$. Consider the solution which comes closely to $\Omega_+$ at some point $x^*$. Represent it in the form $\Omega(x) = \Omega_+ -\epsilon + \delta \Omega(x)$, $\delta\Omega(x^*) = 0$, where $\epsilon$ is a small deviation from $\Omega_+$ at $x = x^*$. Substitute it into Eq.~(\ref{supereq}) and expand $\sin \Omega(x)$ via smallness of $\delta\Omega(x)$ at the vicinity of $x^*$:
		\begin{equation}
		\delta\Omega'(x) \approx 2 U(x) + 2 k \cdot  (\delta \Omega(x) - \epsilon),
		\end{equation}
		where we accounted that $p_y \cos \Omega_+ = -k$, $k=\sqrt{p_y^2 - E^2}>0$.
		The solution with the appropriate boundary condition is:
		\begin{eqnarray}
		\delta\Omega(x) = 2 \int \limits_{x^*}^x U(x') e^{2 k (x - x')} \, dx' + \nonumber \\
		+\epsilon \cdot (1-e^{2k(x-x^*)}).
		\label{variations}
		\end{eqnarray}
		
		In the region $x > x^*$ both terms in~(\ref{variations}) give exponential divergence at $x \to + \infty$ ($x-x' \ge 0$ under the integral). So, the solution which approaches $\Omega_+$ (up to some arbitrarily small value $\epsilon$) runs away exponentially. It proves the statement that $\Omega_+$ are repellors at $x \to +\infty$.
		
		In the region $x < x^*$, $\delta \Omega(x) \to \epsilon$ exponentially fast ($x-x' \le 0$ under the integral) when $x\to -\infty$ and hence $\Omega(x) \to \Omega_+$.  It proves that $\Omega_+$ are attractors at $x \to - \infty$.
		
		 We can prove the statement for $\Omega_-$ in {\bf c} by analogy. For this, we just notice the change of sign in exponents because $p_y \cos \Omega_- = k$.
		
		We have to remark that we can finely adjust the constant $\epsilon$ to cancel out the exponential divergence from the integral part of~(\ref{variations}) at $x \to +\infty$. As we can see below, such solutions indeed exist!
		
		\item[{\bf d}]
		As it follows from {\bf c}, asymptotes $\Omega_+$ ($\Omega_-$) are unstable at $x \to +\infty$ ($x \to -\infty$). However, we require the solutions to satisfy one of the initial conditions~(\ref{separatrix}). We call such solutions left and right separatrices because they separate all solutions by regions. For example, the separatrix $\Omega_r$ separates solutions which are above and below its value $\Omega_+$ at $+\infty$ according the fact that $\Omega_+$ is a repellor at $+\infty$.
		
		Let us demonstrate that once we fixed one of the conditions~(\ref{separatrix}) it defines the only solution. To be more specific, consider $\Omega_r(x)$. To demonstrate the existence of such solution we need to set $x^* = +\infty$ and $\epsilon = 0$ in the previous item. Then $\Omega_r(x) = \Omega_+ + \delta\Omega_r(x)$ where at $x\to +\infty$ we can write by analogy with~(\ref{variations})
		$$
		\delta\Omega_r(x) = 2 \int \limits_{+\infty}^x U(x') e^{2 k (x - x')} \, dx',
		$$
		where $\delta\Omega_r(x) \to 0$ at $x\to +\infty$ which proves the existence of the solution. To show its uniqueness, we suppose two solutions with the same condition $\Omega_{1, 2}(x) \to \Omega_+$ at $x \to +\infty$ and consider its difference $\delta\Omega = \Omega_2 - \Omega_1$ which continuously tends to zero at $x\to +\infty$. While $\delta\Omega$ is small it satisfies the equation:
		$$
		\delta \Omega' = -2 p_y \cos\Omega_1(x) \cdot \delta\Omega
		$$
		with solution:
		$$
		\delta\Omega(x) = \delta\Omega(x_0)\cdot e^{-2 p_y \int\limits_{x_0}^x \cos \Omega_1(x') \, dx'},
		$$
		where $x \le x_0 \to +\infty$. While $x_0$ is fixed we use the limit relation $p_y \cos\Omega_1(x) \to -k$ at $x \to +\infty$ which exposes the exponential divergence at any non-zero $\delta\Omega(x_0)$, ergo $\delta\Omega(x) \equiv 0$.
		
		It should be emphasized that the uniqueness of solutions with the conditions~(\ref{separatrix}) is not valid if $E = \pm p_y$ since $k = 0$.
		
		\item[{\bf e}] Compare now the boundary conditions~(\ref{asymptotes}) for solutions that correspond to physical states with initial conditions~(\ref{separatrix}) for two families of separatrices. The physical solution must fulfill both conditions which is possible only when two separatrix families merge. Thence, the bound state problem is equivalent to the degeneracy of separatrices of Eq.~(\ref{supereq}).
		
	 Notice that the physical solutions are stated by degenerated separatrices, and the corresponding parameter $E$ when the degeneracy occurs is the discrete energy level in a given potential $U(x)$.
	\end{itemize}
\end{proof}

Remark that we denote as $\Omega_l$, $\Omega_r$ the whole families of separatrices. If we need some particular function from a family, we indicate the dependence from $x$: $\Omega_l(x)$, $\Omega_r(x)$. Again, we use notations $\Omega_+$, $\Omega_-$ to describe the whole families of attractors and repellors if we do not indicate explicitly some particular point from these families.

\begin{lemma}[of boundedness]
	Let $U(x)\to 0$ at $x \to \pm \infty$. Let the primitive integral $f_{x_0}(x) = \int_{x_0}^{x} U(x') \, dx'$ of the potential $U(x)$ be a continuous function and the limit $\lim\limits_{x\to \pm \infty} f_{x_0}(x)$ exists (maybe, infinite). Then:
	
\begin{itemize}
	\item[{\bf a}]
	Any solution of Eq.~(\ref{supereq}) is a bounded function for any parameter $E \in (-p_y, p_y)$.
	\item[{\bf b}]
	If $|\lim\limits_{x\to \pm \infty} f_{x_0}(x)| < \infty$, then all solutions of Eq.~(\ref{supereq}) are bounded functions for any parameter $E \in [-p_y, p_y]$.
\end{itemize}	
\end{lemma}
\begin{proof}[{\bf Proof}]
	\begin{itemize}
	\item[{\bf a}]
	First, consider the situation when $k \ne 0$ or $E \in (-p_y, p_y)$.
	
	Continuity of $f_{x_0}(x)$ results in $\Omega(x)$ being a continuous function as to {\bf Lemma 1}. 
	Suppose that $\Omega(x)$ diverges at $+\infty$. From continuity, we always can find an arbitrarily large positive $x_0$ where $p_y \cos \Omega(x_0) = k >0$. We expand $\Omega(x)$ at the vicinity of $x_0$: $\Omega(x) = \Omega(x_0) + \delta \Omega(x)$.
	Up to the first order of $\delta\Omega$ we have:
	\begin{equation}
	\delta\Omega'(x) = 2 U(x) - 2 k \cdot \delta\Omega(x),
	\label{lemma3pr}
	\end{equation}
	which yields the solution:
	\begin{equation}
	\delta\Omega(x) = 2\int \limits_{x_0}^x U(x') e^{-2 k (x - x')} \, dx'.
	\label{sollem3}
	\end{equation}
	We clearly see that $\delta\Omega(x)$ converges at $x \to +\infty$ even at arbitrarily small $k > 0$. Hence, we arrived at the contradiction with our initial assumption of the unboundedness of $\Omega(x)$ at $+\infty$.
	
	By analogy, one can prove the boundedness of any solution of Eq.~(\ref{supereq}) at $x \to -\infty$. Here we will choose an arbitrary large negative $x_0$ where $p_y \cos \Omega(x_0) = -k$.
	
	Notice that $\delta\Omega(+\infty) = 0$; we integrate Eq.~(\ref{lemma3pr}) and substitute~(\ref{sollem3}) into the right-hand side. It yields:
	\begin{eqnarray*}
			&&\int\limits_{x_0}^{+\infty} \delta\Omega(x) \, dx = \\ 
			&&= 2 \int\limits_{x_0}^{+\infty} \int \limits_{x_0}^x U(x') e^{-2k(x-x')} \, dx\, dx'=\frac{f_{x_0}(+\infty)}{k}.
	\end{eqnarray*}
	On the other hand, the direct integration of Eq.~(\ref{lemma3pr}) results in:
	$$
		\delta\Omega(+\infty) = 2 f_{x_0}(+\infty) - 2k \int\limits_{x_0}^{+\infty} \delta\Omega(x) \, dx.
	$$
	Hence, $\delta\Omega(+\infty) = 0$ or $\Omega(+\infty) = \Omega(x_0)$. This result is not surprising because we intentionally chose $x_0$ in that way to satisfy $\Omega(x_0) = \Omega_-$ which is attractor at $x \to +\infty$.

	\item[{\bf b}]
	If $f_{x_0}(x)$ has finite limits at $x \to \pm \infty$, one can show that solutions of Eq.~(\ref{supereq}) are bound on the closed interval $E \in [-p_y, p_y]$. To show this, we need to check what happens on the boundaries of the continuum when $E= \mu p_y$, $\mu = \pm 1$, $k=0$. 
	
	As in item {\bf a}, we assume that $\Omega(x)$ diverges at $x \to +\infty$, thus, we can write $\Omega(x) = \Omega(x_0) + \delta\Omega(x)$, $\sin\Omega(x_0) = \mu$ where $x_0$ can be an arbitrarily large positive number. In Eq.~(\ref{lemma3pr}) we omitted summands of order $\delta\Omega^2$ and higher because $k\ne0$. In this case we have to account for the first non-zero term that is quadratic in $\delta\Omega$:
	\begin{equation*}
	\delta\Omega'(x) = 2 U(x) - \mu p_y \delta\Omega^2(x).
	\end{equation*}
	This equation resembles that of a non-relativistic limit with zero non-relativistic energy.
	
	There are three possible scenarios of the behavior at $+\infty$. The first one, $\delta \Omega^2(x) \sim U(x)$, $x \to +\infty$, gives explicit convergence of $\delta\Omega$ since $U(x) \to 0$, $x \to +\infty$. The second one corresponds to $\delta\Omega^2(x) \sim \delta\Omega'(x)$ which provides the convergence $\delta\Omega \sim 1/x$. The last situation is $\delta\Omega'(x) \sim U(x)$ which gives the convergence if and only if $f_{x_0}(x)$ converges at infinity.
	
	Hence, any solution of Eq.~(\ref{supereq}) is bounded at any parameter $E \in [-p_y, p_y]$ as soon as $f_{x_0}(x)$ is continuous and converges at infinity. 
	\end{itemize}
\end{proof}

As it can be seen from {\bf Lemma 2}, we are interested in the separatrix solutions because only these solutions are related to physical ones. For all further discussions we choose the family of left separatrices $\Omega_l$. We are going to show that the total variance:
$$
\Delta\Omega_l(E) = \Omega_l(+\infty) - \Omega_l(-\infty)
$$
as a function of energy contains the full information of the discrete spectrum. It is stated in the following

\begin{Lev}[Levinson]
Let $f_{x_0}(x)$ be a continuous function which converges at infinity, $E \in [-p_y, p_y]$. Then:

\begin{itemize}
	
	\item[{\bf a}] 
	$\Delta \Omega_l(E)$ is a bounded function on the interval $E \in [-p_y, p_y]$.
	
	\item[{\bf b}] 
	$\Delta \Omega_l(E)$ is a multiple of $2 \pi$ if $E \notin \Spec(U, p_y)$, $\Spec(U, p_y)$ is a discrete specter of $U(x)$ at given $p_y$.
	
	\item[{\bf c}] 
	Any $E \notin \Spec(U, p_y)$ is a point of continuity of $\Delta \Omega_l(E)$.
	
	\item[{\bf d}]
	$\Delta \Omega_l(E)$ has finite jumps of $-2 \pi$ at every point $E_d \in \Spec(U, p_y)$:
	\begin{equation}
		\Delta\Omega_l(E_d+0)-\Delta\Omega_l(E_d -0) = -2\pi.
		\label{djump}
	\end{equation}

	\item[{\bf e}]
	The total number $N_d(p_y)$ of discrete levels of $U(x)$ at any given $p_y > 0$ is defined by:
	\begin{equation}
		\label{Totallevels}
		N_d(p_y)=\frac{\Delta\Omega_l(-p_y)-\Delta\Omega_l(p_y)}{2 \pi}.
	\end{equation}
	
\end{itemize}
\end{Lev}

\begin{proof}[{\bf Proof}]
\begin{itemize}

\item[{\bf a}]	
We know from {\bf Lemma 3} that, under conditions of the theorem, $\Omega_l(x)$ is a bounded function on $x \in (-\infty, \infty)$ at any parameter $E \in [-p_y, p_y]$. In other words $\Delta \Omega_l(E)$ is finite for any $E \in [-p_y, p_y]$ or $\Delta \Omega_l(E)$ is bounded function of $E$.

\item[{\bf b}]
According to {\bf Lemma 2, e)}, two families $\Omega_l$, $\Omega_r$ of separatrices merge if and only if the parameter $E$ corresponds to some discrete energy level. Let $E \notin \Spec(U, p_y)$. Therefore $\Omega_l$ and $\Omega_r$ are disjoint families; $\Omega_l(x)$ starts from some $\Omega_-$ at $x=-\infty$ and comes to, perhaps, some other $\Omega_-$ from the family at $x=+\infty$. Otherwise $\Omega_l(x)$ must tend to $\Omega_+$ at $+\infty$ resulting in $\Omega_l(x) = \Omega_r(x)$ which violates our assumption that $E \notin \Spec(U, p_y)$. Hence, $\Delta\Omega_l(E)$ is a multiple of $2 \pi$.

\item[{\bf c}]
Let $E \notin \Spec(U, p_y)$ where it is natural to assume that $\Spec(U, p_y)$ is a discrete set. Then some $\delta$-vicinity of $E$ is disjoint with $\Spec(U, p_y)$,  $\delta > 0$. Let us consider how $\Omega_l(x, E)$ changes with small variation of the parameter $E$:
$$
\delta\Omega_l(x, E, \epsilon) = \Omega_l(x, E+\epsilon) - \Omega_l(x, E),
$$
where small $0 <|\epsilon| < \delta$. In contrast with the previous consideration where $E$ was fixed, we indicate here $E$ among variables of functions. Subtracting Eq.~(\ref{supereq}) for $\Omega_l(x, E+\epsilon)$ and $\Omega_l(x, E)$, we arrive at the equation for the variation function:
\begin{equation}
\delta\Omega_l' \approx - 2 \epsilon - 2 p_y \cdot \cos \Omega_l(x, E) \cdot \delta\Omega_l.
\label{leftsep}
\end{equation}
Remark that the initial condition depends on $\epsilon$ because:
\begin{equation}
\delta \Omega_l(-\infty, E, \epsilon) \! = \! \Omega_-(E+\epsilon) \! - \! \Omega_-(E) \!\approx \!-\frac{\epsilon}{k}.
\label{incond}
\end{equation}
The solution reads:
\begin{eqnarray}
&&\delta \Omega_l(x, E, \epsilon) = \nonumber\\
&& = - 2 \epsilon \int \limits_{-\infty}^x e^{2 p_y \int_x^y \cos \Omega_l(y', E) \, dy'} \, dy.
\label{varepsilon}
\end{eqnarray}

First, let's demonstrate that~(\ref{varepsilon}) meets the initial condition~(\ref{incond}).
According to~(\ref{separatrix}), we may approximate $p_y \cos\Omega_l(y', E) \to p_y \cos\Omega_- = k$ at $x \to -\infty$ because $y \le y' \le x$. Hence, at $x \to -\infty$ we see that:
$$
\delta \Omega_l(-\infty, E, \epsilon) \to  - 2 \epsilon \int \limits_{-\infty}^x e^{2 k (y -x)} \, dy = - \epsilon/k.
$$

Now we are ready to show the convergence of~(\ref{varepsilon}) at $+\infty$ and that $\delta\Omega_l(+\infty, E, \epsilon) = -\epsilon/k$. First, divide~(\ref{varepsilon}) into two parts: the first part is the $y$-integral where $-\infty < y < x_0$, the second part is  the $y$-integral where $x_0 < y < x$. $x_0 < x$ is big positive number such that we can use the approximation $p_y \cos \Omega_l(y', E) \approx p_y \cos \Omega_- = k$ while $y' > x_0$. The first part can be estimated at $x \to +\infty$ as follows:
\begin{eqnarray*}
&&-2 \epsilon \int\limits_{-\infty}^{x_0} e^{2 p_y (\int_x^{x_0} +\int_{x_0}^y)\cos \Omega_l(y', E) \, dy'} \, dy \approx \\
&&-2 \epsilon \int\limits_{-\infty}^{x_0} e^{2 p_y \int_{x_0}^y\cos \Omega_l(y', E) \, dy'} \, dy \cdot e^{-2 k (x - x_0)} = \\
&& = \delta\Omega_l(x_0, E, \epsilon) \cdot e^{-2 k (x - x_0)} \to  0.
\end{eqnarray*}
The second part gives the desirable limit $\delta\Omega_l(+\infty, E, \epsilon)$:
\begin{eqnarray*}
	&&-2 \epsilon \int\limits_{x_0}^x e^{2 p_y \int_x^y\cos \Omega_l(y', E) \, dy'} \, dy \approx \\
	&&\approx -2 \epsilon \int\limits_{x_0}^x e^{2 k (y - x)} \, dy \to  -\frac{\epsilon}{k}.
\end{eqnarray*}
Hence, $\delta\Omega_l(+\infty, E, \epsilon) = \delta\Omega_l(-\infty, E, \epsilon) = -\epsilon/k + O(\epsilon^2)$. We remark the equality of values of $\delta\Omega_l$ at $\pm \infty$ not just up to order of $\epsilon^2$ because we have proven here that the difference tends to zero with $\epsilon$. But according to item {\bf b} of this theorem, the difference must be a multiple of $2 \pi$ whence the only one opportunity is possible. Finally, we conclude that:
\begin{eqnarray*}
&&\Delta\Omega_l(E+\epsilon) - \Delta\Omega_l(E) = \\
&&= \delta\Omega_l(+\infty, E, \epsilon) - \delta\Omega_l(-\infty, E, \epsilon) = 0.
\end{eqnarray*}

Hence, we proved that any $E \notin \Spec(U, p_y)$ is the point of continuity of the function $\Delta\Omega_l(E)$. {\bf We also proved that } $\Delta\Omega_l(E)$ {\bf is a piecewise-constant function with only possible discontinuity points from} $\Spec(U, p_y)$.

We emphasize that the statement of this item is true even for the boundaries of continuum where $E = \pm p_y$ since $E = \pm p_y$ are not limit points of $\Spec(U, p_y)$ (see the {\bf Remark 1}). For example, for $E = p_y$ we take 
$$
\delta \Omega_l(x, E=p_y, \epsilon) = \Omega_l(x, p_y - \epsilon) - \Omega_l(x, p_y),
$$
where $\epsilon \approx k^2/(2 p_y) \to +0$. Then the condition~(\ref{incond}) is valid because $\epsilon/k \approx k/(2 p_y) \to 0$. 

\item[{\bf d}]
Now we understand the behavior of $\Delta\Omega_l(E)$ when $E \notin \Spec(U, p_y)$. In this item we consider the situation when $E = E_d \in \Spec(U, p_y)$ where we assume that $\Spec(U, p_y)$ is a discrete set or each element is an isolated point. As it follows from {\bf Lemma 2, e)}, two separatrix families merge when $E=E_d$. We call these merged separatrices as $\Omega_d$ family. 

$E_d$ is an isolated point of $\Spec(U, p_y)$. Then $\delta > 0$ exists such that $\delta$-vicinity of $E_d$ does not contain any other points from $\Spec(U, p_y)$ except $E_d$. Let us consider the variation function:
$$
\delta\Omega_l(x, E_d, \epsilon) = \Omega_l(x, E_d+\epsilon) - \Omega_d(x, E_d),
$$
where $\epsilon$ can be arbitrarily small, $0 <|\epsilon| < \delta$ . Afterwards, we repeat the procedure from item {\bf c} of the theorem which gives exactly the same initial condition~(\ref{incond}) and in Eq.~(\ref{leftsep}) we need to substitute $\Omega_l(y', E) \to \Omega_d(y', E_d)$. Thence the approximate solution for $\delta\Omega_l(x, E_d, \epsilon)$ reads:
\begin{eqnarray}
&&\delta \Omega_l(x, E_d, \epsilon) = \nonumber\\
&& = - 2 \epsilon \int \limits_{-\infty}^x e^{2 p_y \int_x^y \cos \Omega_d(y', E_d) \, dy'} \, dy.
\label{vardepsilon}
\end{eqnarray}

But analysis of Eq.~(\ref{vardepsilon}) at $x \to +\infty$ gives different result from those of Eq.~(\ref{varepsilon}). The reason is that $\Omega_d(x, E_d)$ comes to $\Omega_+$ at $x \to +\infty$ as per the conditions~(\ref{asymptotes}). This gives $p_y \cos \Omega_d(+\infty, E_d) = p_y \cos\Omega_+ = -k$d which results in the exponential divergence of $\delta\Omega_l(x, E_d, \epsilon)$ at $x \to +\infty$ for any $|\epsilon| > 0$. Formally, this divergence indicates instability of the solution $\Omega_d(x, E_d)$ towards infinitely small variations from the parameter $E_d$. This conclusion is already obvious because we know that at $E = E_d + \epsilon$ we have two disjoint families of separatrices and our separatrix $\Omega_l$ tends to $\Omega_-$ at $x \to +\infty$.

The non-trivial conclusion which can be drawn from~(\ref{vardepsilon}) is that:
\begin{equation}
\sign\left( \delta \Omega_l \right)= - \sign( \varepsilon).
\label{signs}
\end{equation} 
We are going to show that it leads to~(\ref{djump}).

We can use the approximate solution~(\ref{vardepsilon}) at the region $x < R$ if the condition $\delta\Omega_l(x<R, E_d, \epsilon) \ll 1$ is met. Fix some small value of $\delta\Omega_l$:
$$
\delta\Omega_l(R, E_d, \epsilon) \equiv \alpha.
$$
It means that $R$ is a function of two parameters $\alpha$ and $\epsilon$ and $R(\alpha,\epsilon) \to +\infty$ at fixed $\alpha$ and $\epsilon \to 0$. Introduce the following variance: 
$$
 \delta\Omega_d=\Omega_d(R(\alpha, \epsilon), E_d) - \Omega_+,
$$
where $\delta\Omega_d \to 0$ at $R \to +\infty$.
Finally, we have for the left separatrix:
\begin{eqnarray*}
&&\Omega_l(R(\alpha,\epsilon), E_d+\epsilon) = \Omega_+ + \delta\Omega_d + \alpha,
\end{eqnarray*}
where $\alpha$ is fixed and $\delta\Omega_d \to 0$ at $\epsilon \to 0$ or equivalently:
$$
\Omega_l(R(\alpha,\epsilon), E_d+\epsilon) \to \Omega_+ + \alpha
$$
at $\epsilon \to 0$ and arbitrarily small but fixed $\alpha$. According to the definition of $\alpha$ and Eq.~(\ref{signs}), we get 
$$
\sign(\alpha) = -\sign(\epsilon).
$$
It means that at $\epsilon > 0$ ($\epsilon < 0$) the left separatrix $\Omega_l(R, E_d+\epsilon) < \Omega_+$ ($\Omega_l(R, E_d+\epsilon) > \Omega_+$) at $R \to +\infty$ and ergo $\Omega_l(x, E_d+\epsilon)$ falls onto the asymptote $\Omega_-$ which is right under (above) the asymptote $\Omega_+ = \Omega_d(+\infty, E_d)$.
Thence:
$$
\Omega_l(+\infty, E_d + 0) - \Omega_l(+\infty, E_d - 0) = - 2 \pi
$$
or equivalently:
$$
\Delta\Omega_l(E_d+0)-\Delta\Omega_l(E_d -0) = -2\pi.
$$
We used the fact that here $\Omega_l(-\infty, E_d+0)= \Omega_l(-\infty, E_d-0)$.

One can show by analogy that the right separatrix experiences jumps with the same sign:
$$
\Delta\Omega_r(E_d+0)-\Delta\Omega_r(E_d -0) = -2\pi.
$$
In this sense, the right separatrix does not give any additional information about the discrete spectrum.

\item[{\bf e}]
We proved that the function $\Delta\Omega_l(E)$ is a bounded piecewise-constant function which experiences final jumps of $-2 \pi$ at every point $E_d$ of discrete spectrum of the confinement $U(x)$. $\Delta\Omega_l(E)$ is continuous at any other points where $E \notin \Spec(U, p_y)$.

It allows us to calculate the total number of discrete levels as the difference of $\Delta\Omega_l(E)$ on the ends of the interval $[-p_y, p_y]$ which immediately gives Eq.~(\ref{Totallevels}).

However, we understand $\Delta\Omega_l(\pm p_y)$ only in the sense of the limit relation $\Delta\Omega_l(\pm p_y) = \lim \limits_{\epsilon \to +0} \Delta\Omega_l(\pm (p_y- \epsilon))$ because separatrices are not well defined at the boundaries of the continuum as to {\bf Lemma 2}.

\end{itemize}
\end{proof}

\begin{remark1}[for the Levinson Theorem]
We need to remark that assumptions made in the head of the Levinson theorem provide that $\Spec(U, p_y)$ is discrete set. Indeed, assume that $\Spec(U, p_y)$ has one limit point $E_0 \in [-p_y, p_y]$. It means that infinitesimal vicinity of this point contains an infinite number of isolated points from $\Spec(U, p_y)$. But for any isolated point, the item {\bf d} of the theorem is valid which leads to $\Delta\Omega_l(E \to E_0) \to \infty$; this contradicts with the item {\bf a} of the theorem of boundedness of this function for any $E \in [-p_y, p_y]$. Hence, $\Spec(U, p_y)$ does not contain limit points.
\end{remark1}

\begin{remark1}[for the Levinson Theorem]
	Even if $|\lim \limits_{x\to\pm \infty} f_{x_0}(x)| = \infty$, all proofs and statements of the Theorem are valid for open interval $E \in (-p_y, p_y)$ because $k = \sqrt{p_y^2 - E^2} > 0$. However, at least one of the points $E = \pm p_y$ is limit point of $\Spec(U, p_y)$ which makes $\Delta\Omega_l(E)$ unbound on the closed interval $E\in [-p_y, p_y]$.
\end{remark1}

\begin{remark1}[for the Levinson Theorem]
	One can get the number of discrete levels between any two given energies $|E_{1, 2}| \le p_y$, $E_{1, 2} \notin \Spec(U, p_y)$:
	\begin{equation}
	\label{Numlevels}
	N_d(p_y, E_1, E_2)=\left|\frac{\Delta\Omega_l(E_2)-\Delta\Omega_l(E_1)}{2 \pi} \right|.
	\end{equation}
\end{remark1}

Hence, the function $\Delta \Omega_l(E)$ plays the same role as the scattering phase in the non-relativistic theory. In other words, the theorem represents the relativistic Levinson theorem for the 2D Dirac equation with the 1D potential.
\\

{\it Example for $\delta$-potential}

Finally, we give an example for the simple case of the $\delta$-potential $U(x) = G \cdot \delta(x)$. Let us demonstrate that the total number of discrete levels $N_d(p_y) = 1$ at any $p_y \ne 0$ and $G \ne \pi n$, $n$ is integer, $N_d$ is defined by Eq.~(\ref{Totallevels}). We need to consider Eq.~(\ref{supereq}) only at $E= \pm p_y$.

All solutions of Eq.~(\ref{supereq}) are constructed from solutions of the free motion equation~(\ref{freemotion}) separately at $x<0$ and $x>0$ with the matching condition
\begin{equation}
	\Omega(+0) = \Omega(-0) + 2 G.
\end{equation}

We first analyze the solutions of Eq.~(\ref{freemotion}).
If $E = p_y$, then we have $\Omega'(x) = -2 p_y (1 + \sin \Omega) \le 0$ and $\Omega'(x) = 0$ only for the case of stationary points $\Omega_0 \equiv \Omega_\pm = -\pi/2 + 2 \pi n$. Hence, all non-stationary solutions of Eq.~(\ref{freemotion}) decrease strictly monotonically from some stationary point $\Omega_0 + 2\pi$ at $x = -\infty$ to $\Omega_0$ at $x = +\infty$. Notice that two families of stationary points merge at $E = \pm p_y$. 

In the case $E = - p_y$ all non-stationary solutions of Eq.~(\ref{freemotion}) increase strictly monotonically from some stationary point $\Omega_0 - 2\pi$ at $x = -\infty$ to $\Omega_0$ at $x = +\infty$. 

Represent the confinement strength in the following form:
$$
G = \pi (n_G + \delta n_G),
$$ 
where $n_G$ is integer and $\delta n_G \in(0, 1)$. Then:
$$
\Omega_l(x < 0, \pm p_y) = \Omega_-(\pm p_y)
$$ 
and 
$$\Omega_l(+0, \pm p_y) = \Omega_-(\pm p_y) + 2 \pi n_G + 2 \pi \delta n_G,
$$
where $\Omega_0 = \Omega_- + 2 \pi n_G$ is stationary point and $2 \pi \cdot \delta n_G \in (0, 2\pi)$ which means that $\Omega_l(x, \pm p_y)$ at $x>0$ comes along some non-stationary solution which decreases (increases) at $E=p_y$ ($E = -p_y$), ergo $\Omega_l(+\infty, p_y) = \Omega_0$ ($\Omega_l(+\infty, -p_y) = \Omega_0 + 2\pi$) at $E=p_y$ ($E = -p_y$). Equivalently, $\Delta\Omega_l(p_y) = 2 \pi n_G$ and $\Delta\Omega_l(-p_y) = 2 \pi n_G + 2 \pi$. Hence, $N_d(p_y) = 1$.

\section{Geometrical interpretation of the relativistic Levinson theorem}

The problem of bound states in graphene stripes can be analyzed similarly to what happens in mechanical autonomous systems. Let us consider the following system of equations:
\begin{equation}
\label{Autosys}
\left\{
\begin{array}{ll}
U'(x) = G(U) \\
\Omega'(x)=2 \left(U(x)- E\right) -2 p_y \sin {\Omega(x)},
\end{array}
\right.
\end{equation}
where the second equation here is just Eq.~(\ref{supereq}). We may consider that Eq.~(\ref{Autosys}) represents integral curves of some vector field
$$
{\bf F}(U, \Omega) = { G(U) \choose 2 \left(U - E\right) -2 p_y \sin {\Omega} },
$$
whereas the coordinate $x$ is just some parametrization of these curves. Though the system~(\ref{Autosys}) is not Hamiltonian as in usual mechanics, it is still an autonomous system of differential equations and, therefore, it can be analyzed in terms of the phase trajectories in so-called phase space $\mathfrak{D}$. In our case, the phase space $\mathfrak{D}$ is the $(U, \Omega)$-stripe:
$$
\mathfrak{D} = \{(U, \Omega)| U \in [\inf\limits_{x \in \mathbb R} U(x), \sup\limits_{x \in \mathbb R} U(x)], \Omega\in \mathbb R \},
$$
where $\mathbb R = (-\infty, +\infty)$.

However, our system~(\ref{Autosys}) is more complicated than usual autonomous systems. To see this, notice that the function $G(U)$ is different for each interval of monotonicity $I_j = [x_{j-1}, x_j]$ of $U(x)$. It means that we have different maps for each $I_j$ and we need to match these maps continuously. In other words, instead of one autonomous system we have the whole chain of systems:
\begin{equation}
{\bf F}_j(U, \Omega) = {U'(x) \choose \Omega'(x)} = {G_j(U) \choose 2 \left(U - E\right) -2 p_y \sin {\Omega}}
\label{jtheq}
\end{equation}
which are autonomous on the corresponding intervals $I_j$, $x\in I_j$ is some parametrization, and ${\bf F}_j(x_j) = {\bf F}_{j+1}(x_j)$. All trajectories of the field ${\bf F}_j$ fill the whole stripe:
$$
\mathfrak{D}_j = \{(\Omega , U)| U \in [\inf\limits_{x \in I_j} U(x), \sup\limits_{x \in I_j} U(x)], \Omega\in \mathbb R \}.
$$

Let us formulate the following
\begin{lemma}[of stationary points]
	Let $U(x) \in C^1$ have a finite number $N$ of monotonicity intervals $I_j = [x_{j-1}, x_j]$, $x_0 = -\infty < x_1 < \dots < x_{N-1 }< x_N = +\infty$. Let $U(x)$ be a strictly monotonic function on each $I_j$. Let $U(x) \to 0$ at $x \to \pm \infty$. Then:
	\begin{itemize}
		\item[{\bf a}]
		$U'(x) \to 0$ at $x \to \pm \infty$.
		\item[{\bf b}]
		Functions $G_j(U)$ are definite on corresponding intervals $I_j$, $j = 1, \dots, N$ and $G_1(0) = G_N(0) = 0$.	
		\item[{\bf c}]	
		The number of stationary points of $j$-th Eq.~(\ref{jtheq}) is exhausted by the following series:
		$$
		\left( U_\sigma, \arcsin \left(\frac{U_\sigma - E}{p_y} \right) + 2\pi n \right)
		$$
		or
		$$
		\left( U_\sigma, \pi -\arcsin \left(\frac{U_\sigma - E}{p_y} \right) + 2 \pi n\right)
		$$
		where $n$ is integer, $|U_\sigma - E| \le p_y$ and $G_j(U_\sigma)=0$. 
	\end{itemize}
\end{lemma}

\begin{proof}[{\bf Proof}]
	\begin{itemize}
		\item[{\bf a}]
		It is straightforward from the monotonic behavior of $U(x)$ at infinity and $U(x) \to 0$ at $x \to \infty$.
		
		\item[{\bf b}]
		$U(x)$ is strictly monotonic on each $I_j$, therefore an inverse function exists: $x_j(U)$. Thereby we get $G_j(U) = U'(x_j(U))$. 
		
		We know that $I_1 = (-\infty, x_1]$, $I_N = [x_{N-1}, +\infty)$ and $U'(x) \to 0$ at $x \to \pm \infty$ where $U(x) \to 0$. It immediately yields: $G_1(0) = \lim\limits_{x\to - \infty} U'(x) = 0$ and $G_N(0) = \lim\limits_{x\to +\infty} U'(x) = 0$.
		
		\item[{\bf c}]
		This statement follows from the solution of the equation: $$ {\bf F}_j (U, \Omega) = 0.$$
		
	\end{itemize}
\end{proof}

Further we call the whole chain of connected maps for ${\bf F}_j(U, \Omega)$ as ${\bf F}(U, \Omega)$ where each trajectory from $\mathfrak{D}$ corresponds to some solution of Eq.~(\ref{Autosys}). The properties of these trajectories are formulated in the
\begin{Poincare}[of Poincare indeces]
	Let all restrictions of {\bf Lemma 4} be valid. Let us consider the following mapping 
	$\mathfrak{D} \to \mathfrak{R}$ by the rule:
	\begin{equation}
	\left\{
	\begin{aligned}
		X(U, \Omega) = (U + a \cdot p_y) \cos\Omega, \\
		Y(U, \Omega) = (U + a \cdot p_y) \sin\Omega,
	\end{aligned}	
	\right.
	\label{map}
	\end{equation}
	where $+\infty > a \cdot p_y > -\inf \limits_{x \in \mathbb R} U(x)$ is some parameter, $E \in (-p_y, p_y)$, $E \notin \Spec(U, p_y)$. Then:
	
	\begin{itemize}
		\item[{\bf a}]
		All stable trajectories of the vector field ${\bf P}(X, Y) = {\bf F}(U(X, Y), \Omega(X, Y))$, $(X, Y) \in \mathfrak{R}$
		are open.
		All unstable trajectories (separatrices) are closed.
		
		\item[{\bf b}]
		In the previous section we introduced the total variance $\Delta\Omega_s(E)$, $s$ indicates left or right separatrix. The relation $\Delta\Omega_s(E)/(2 \pi)$ equals to integer number $\mathfrak{p}$ of full rotations of corresponding closed trajectory in the phase space $\mathfrak{R}$:
		$$
			\Delta\Omega_s(E) = 2 \pi \mathfrak{p}_s.
		$$
		 $\mathfrak{p}_s$ is the Poincare index of closed trajectory. 
	\end{itemize}
\end{Poincare}
\begin{proof}[{\bf Proof}]
	\begin{itemize}
		\item[{\bf a}]
		The mapping~(\ref{map}) is the mapping of stripe $\mathfrak{D}$ to the ring $\mathfrak{R}$ where all points $(U, \Omega + 2 \pi n)$, $n$ is integer, are identified. 
		
		The asymptotic behavior of stable trajectories of the field ${\bf P}(X, Y)$ is referred to stable solutions of Eq.~(\ref{supereq}) which start from attractor $\Omega_+$ at $x \to -\infty$ and finish to attractor $\Omega_-$ at $x \to +\infty$ as to {\bf Lemma 2}. Accounting that $U(x) \to 0$ at $x \to \pm \infty$, we conclude that stable trajectories in $\mathfrak{R}$ space start from the point 
		$$
		P_i = (-a \cdot k, -a \cdot E)
		$$
		because $X_i = a \cdot p_y \cos\Omega_+$, $Y_i = a \cdot p_y \sin \Omega_+$;
		and finish by another point
		$$
			P_f = (a \cdot k, -a \cdot E)
		$$
		because $X_f = a \cdot p_y \cos\Omega_-$, $Y_f = a \cdot p_y \sin \Omega_-$. If $E\in (-p_y, p_y)$ then $k > 0 $ and $P_f \ne P_i$. This means that stable trajectories are open.
		
		According to~(\ref{separatrix}), if $E \notin \Spec(U, p_y)$, $\Omega_l$ ($\Omega_r$) starts and finishes on the asymptotes from the same family: $\Omega_-$ for $\Omega_l$ and $\Omega_+$ for $\Omega_r$. Then, $P_i$ and $P_f$ are identical for them or, equivalently, their trajectories in $\mathfrak{R}$ space are closed.
		
		\item[{\bf b}]
		 It follows from the Levinson Theorem that $\Delta \Omega_l(E) = 2 \pi \mathfrak{p}_l$ where $\mathfrak{p}_l$ is integer. But from the continuity of $\Omega_l(x)$ we conclude that $\mathfrak{p}_l$ is the number of full rotations of the closed trajectory corresponding to the separatrix $\Omega_l$ in $\mathfrak{R}$ space. In other words, $\mathfrak{p}_l$ is the Poincare index of this closed trajectory~\cite{Poincare}.
	\end{itemize}
\end{proof}

\begin{figure}[htbp]
	\center
	\includegraphics[width=0.5\textwidth]{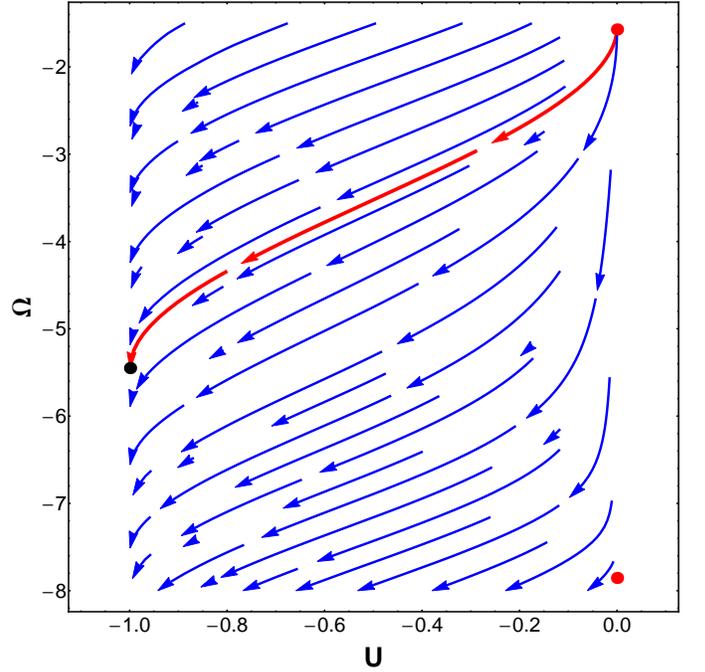}
	\caption{The vector field ${\bf F}(E = p_y)$, $p_y = 0.1$, $U_0 = 1$ on the interval $I_1 = (-\infty, 0)$. The trajectory $(U, \Omega_l(x_1(U)))$ corresponding to the separatrix $\Omega_l(x)$ (red streamline) starts from the initial (red) point $(U = 0, \Omega = -\pi/2)$ and ends when $U = -U_0 = -1$ (black point). The distance between red points is equal to $2 \pi$.} 
	\label{Vec1}
\end{figure}

\begin{figure}[htbp]
	\center
	\includegraphics[width=0.5\textwidth]{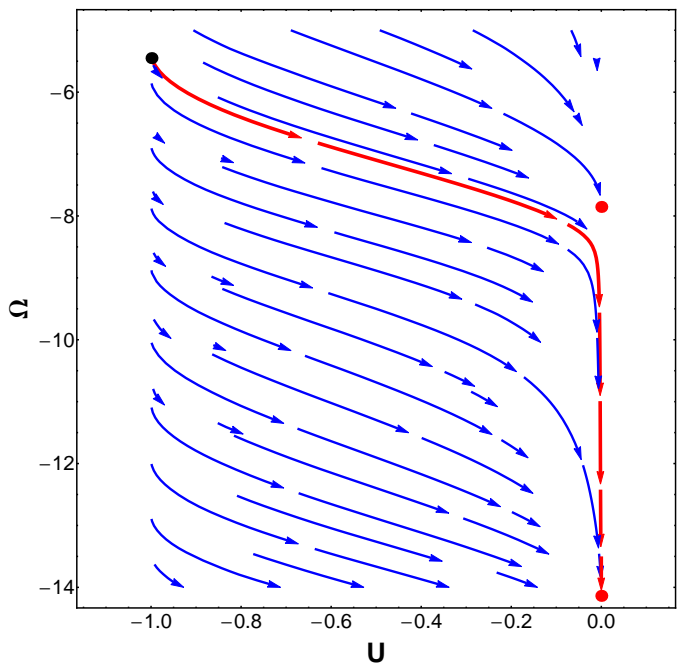} 
	\caption{The vector field ${\bf F}(E = p_y)$, $p_y = 0.1$, $U_0 = 1$ on the interval $I_2 = (0, +\infty)$. The trajectory $(U, \Omega_l(x_2(U)))$ corresponding to the separatrix $\Omega_l(x)$ (red streamline) starts from the black point which provides the continuity of $\Omega_l(x)$ at $x = 0$ and ends at the red point $(U = 0, \Omega = -9 \pi/2)$. The distance between red points is equal to $2 \pi$.
		\label{Vec2}}
\end{figure}

Here we present a simple example of the spectral analysis for the Lorentzian shaped confinement $$
U(x) = - U_0/(x^2+1).
$$
We are going to plot the vector field ${\bf F}(U, \Omega)$ and calculate the number of bound states at some particular $p_y$ and $U_0$.

First, we need to find $G_j(U)$ for each interval of monotonicity $I_1 = (-\infty, 0)$ and $I_2 = (0, + \infty)$:
$$
G_{n}(U) = (-1)^n \frac{2U^2}{U_0} \sqrt{-\frac{U_0}{U} - 1}
$$
for the interval $I_n$, $n = \{1, 2\}$, $U \in [-U_0, 0]$. 

Then we set the parameters $p_y = 0.1$, $U_0 = 1$. In order to find the total number of confined modes, we apply Eq.~(\ref{Totallevels}). We need to plot the phase portrait only for two energies $E = \pm p_y$. Pictures (Fig.~\ref{Vec1}--\ref{Vec2}) of the vector field ${\bf F}(E = p_y)$ show the approximate trajectory $(U, \Omega_l(x(U)))$ (red line) for two intervals $I_{1, 2}$. We chose the point $(U = -10^{-6}, \Omega = - \pi/2 + 0.05)$ as the initial condition for the trajectory $(U, \Omega_l(x_1(U)))$ on the interval $I_1$.  Matching trajectories corresponding to the intervals $I_1$ and $I_2$ (black points on Fig.~\ref{Vec1}--\ref{Vec2}) we finally obtain the variance $\Delta \Omega_l(p_y) = -4 \pi$. Analogically, drawing such pictures for $E = -p_y$ we get $\Delta \Omega_l(p_y) = 0$. Eq.~(\ref{Totallevels}) yields $N_d(p_y) = 2$ confined energy levels for $p_y = 0.1$.

We have to remark that initial condition for $\Omega_l$ must be perturbed from ideal point $(U=0, \Omega = \Omega_-)$ because it is stationary point of Eq.~(\ref{Autosys}) according to {\bf Lemma 4}. However, the result is stable towards little shaking of initial conditions because of the stability of the Poincare index or so-called topological charge.

\section{Conclusions}
The variable phase method has been developed herein for the electrostatically confined 2D massless Dirac-Weyl particles such as electrons in graphene devices. The desirable phase function $\Omega(x)$ appears as the phase between two chiral states whose superposition yields the wave function of the confined state. Besides the well-known non-relativistic and semi-classical limits, it has been shown that confined states with small $p_y$ (see the condition~(\ref{pycond})) are successfully described in the so-called $\delta$-potential limit that is valid for every integrable potential $U(x)$. The relativistic Levinson theorem has then been formulated and proved for the variance $\Delta \Omega_l(E)$ of the separatrix $\Omega_l(x)$ of Eq.~(\ref{supereq}). As a consequence of the theorem, the number of confined modes with given $p_y$ has been derived. Finally, the geometrical approach to find the function $\Delta \Omega_l(E)$ has been suggested.

We note that this paper is dedicated exceptionally to the discrete part of the specter. The developed approach can be extended to analyze half-bound and quasi-bound states where the last ones are important for better understanding of supercriticality.

\section{Acknowledgements}
I am grateful to M. V. Entin for useful discussions and the critical leading of the manuscript. The work was supported by the RFBR grant 14-02-00593. 

\section{Appendix A: unambiguous solution of the $\delta$-potential}

One can find in the literature that $U(x) = G \delta(x)$ does not have definite solutions for Dirac-Weyl equation~\cite{Calkin}--\cite{McKellar}. This problem arises from the fact that the wave function is discontinuous at $x = 0$ and it results in the ambiguous integral of the type
$$
\int\limits_{-\epsilon}^\epsilon \delta(x) \theta(x) \, dx
$$ 
which takes an arbitrary value from the segment $[0, 1]$, $\theta(x)$ is the Heaviside step function, $\epsilon \to +0$. This problem is bypassed by A. Calogeracos {\it et al.}~\cite{Imagawa}. They represented the wave function $\Psi(x)$ as the $x$-ordered exponent (the analogue of the evolution operator) acting on the wave function in the initial point $x_0$. We cite herein the exact solution of Eq.~(\ref{eq}) in order to demonstrate explicitly the absence of any ambiguities. 

Let us start from Eq.~(\ref{eq}):
\begin{equation}
\label{eqDelta}
g''\left(x\right)+2i\left(E-G \delta(x)\right)g'\left(x\right)-p_y^2 g\left(x\right)=0.
\end{equation}
The function $g(x)$ appears to be continuous, $g'(x)$ is discontinuous at $x = 0$. Assume that $g'(\pm 0) \ne 0$ and divide this equation over the function $g'(x)$, $x \in I_\epsilon = (-\epsilon, \epsilon)$. Integrating then this equation over the interval $I_\epsilon$ and taking the limit $\epsilon \to +0$ we arrive at the correct matching condition:
\begin{equation}
\label{MatchDelta}
\frac{g'(+0)}{g'(-0)} = e^{2 i G}.
\end{equation}

If one is interested in the discrete spectrum of this problem one has to apply the condition~(\ref{MatchDelta}) to the function $g(x) = g_0 e^{-i E x} e^{- k |x|}$ which represents the common form of the continuous at $x = 0$ bounded solution of Eq.~(\ref{eqDelta}), $k = \sqrt{p_y^2 - E^2}$. This yields explicitly the spectrum~(\ref{Deltalimit}). The initial assumption $g'(\pm 0) \ne 0$ is obviously valid for such functions $g(x)$.

If we consider the scattering problem with definite $|E| > p_y$, the continuous function $g(x)$ has the following form:
$$
g(x) = 
\left\{
\begin{array}{ll}
A e^{ix (k - E)} + B e^{-ix (k + E)}, x < 0\\
(A + B) e^{ix (k - E)}, x > 0,
\end{array}
\right.
$$
$k = \sqrt{E^2 - p_y^2}$. Applying the condition~(\ref{MatchDelta}) one can receive the transmission coefficient:
$$
T = \left|1 + \frac{B}{A} \right|^2 = \frac{k^2}{k^2 + p_y^2 \sin^2 G}.
$$
Finally, we have to check that the initial assumption $g'(\pm 0) \ne 0$ is not violated. $g'(+0) \ne 0$ as far as $E \ne k$ when $p_y \ne 0$. Suppose then that $g'(-0) = 0$ which leads to $A (k - E) = B (k + E)$ or equivalently $T = 4 k^2 / (k + E)^2$. This makes no physical sense because the transmission coefficient $T$ is not dependent on the parameter $G$ in this case. Hence, the unambiguous solution for the case of the $\delta$-potential is provided.

We can suggest an easier way to get the discrete spectrum for this potential. By integrating Eq.~(\ref{supereq}) and applying boundary conditions~(\ref{asymptotes}) we finally get:
\begin{equation}
\Delta \Omega = \Omega_+ - \Omega_- = 2 G
\end{equation}
which gives explicitly the spectrum~(\ref{Deltalimit}).

\end{document}